%% file: ms.tex
\documentclass[runningheads]{llncs}
\usepackage{enumitem}   
\usepackage{amsmath}
\usepackage{amsfonts}
\usepackage{extarrows}
\usepackage{listings}
\usepackage{booktabs}
\usepackage{multirow}
\usepackage [table]{xcolor}
\usepackage[english]{babel}
\usepackage{listings}[2015/06/04]  
\lstdefinelanguage{abs}{keywords=
	{module,from,Fut,import,put,get,do,od,data,type,def,case,of,class,interface,extends,implements,inherits,if,else,await,delegate,return,suspend,skip,while,some,new,return,then,fi,var,and,or,inv,post,pre,where,others,h,local,respects,with,Any,op,var,out}, 
	sensitive=true,
	morecomment=[l]{//},
	morestring=[b]"}

\lstnewenvironment{abs}{%
	\lstset{language=ABS,columns=fullflexible,mathescape=true,%
		keywordstyle=\bf\sffamily,commentstyle=\sl\ttfamily,%
		frame=single,
		basicstyle=\sffamily,inputencoding=latin1, 
		extendedchars}
	\csname lst@SetFirstLabel\endcsname}
{\csname lst@SaveFirstLabel\endcsname}

\lstdefinelanguage{abs-nonumber}{keywords=
	{module, from, import,Int,put,get,do,od,data,type,def,case,of,class,interface,extends,implements,inherits,if,else,await,delegate,return,
		suspend,skip,while,some,new,return,then,fi,var,and,or,inv,post,pre,where,others,h,local, respects, with,Any,op,var,out}, 
	sensitive=true,
	morecomment=[l]{//},
	morestring=[b]"}

\lstnewenvironment{abs-nonumber}{%
	\lstset{language=ABS,columns=fullflexible,mathescape=true,%
		keywordstyle=\bf\sffamily,commentstyle=\sl\ttfamily,%
		frame=single,
		basicstyle=\sffamily,inputencoding=latin1, 
		extendedchars}
	\csname lst@SetFirstLabel\endcsname}
{\csname lst@SaveFirstLabel\endcsname}

\lstdefinelanguage{maude}{keywords=
	{sort,op, eq, crl, ceq, vars, var, subsorts, subsort, sorts}, 
	sensitive=true,
	morecomment=[l]{//},
	morestring=[b]"}

\lstnewenvironment{maude}{%
	\lstset{language=Maude,columns=fullflexible,mathescape=true,%
		keywordstyle=\bf\sffamily,commentstyle=\sl\ttfamily,%
		numbers=none,
		basicstyle=\sffamily,inputencoding=latin1, 
		extendedchars}
	\csname lst@SetFirstLabel\endcsname}
{\csname lst@SaveFirstLabel\endcsname}

\usepackage{multicol, blindtext}
\usepackage{mathtools, cuted}
\usepackage{lipsum, color}
\usepackage{amssymb}
\usepackage{booktabs,multirow}
\usepackage{tikz}
\usetikzlibrary{shapes}

\colorlet{shadecolor}{gray!30}

\input{macros} 
\newcommand{\ignore}[1]{}

\begin{document}
\input{title}
\input{background}
\bibliographystyle{plain}
\bibliography{mainref}
\newpage
\input{appendix}
\end{document}
\endinput

%% file: macros.tex
\newcommand{\many}[1]                              {\overline{#1}}

\newcommand{\kw}[1] {\texttt{#1}} 

\newcommand{\dra}{\im{\dra}}

\newcommand{\IFLONG}[1]{}

\newcommand{\transition}[1]{\xrightarrow{}}
\newcommand{\Equationn}[1]{={~}}

\newcommand{\OPRULE}[1]{\textsf{#1}}

\newcommand{\mvar}{\it} 

\def\myttsize{\fontsize{10}{12}}

\newcommand{\mm}[1]{\overline{#1}}

  \newcommand{\Blue}[1]{\textcolor{blue}{#1}}
  \definecolor{airforceblue}{rgb}{0.36, 0.54, 0.66}
  \definecolor{bistre}{rgb}{0.24, 0.17, 0.12}
  \definecolor{brown(traditional)}{rgb}{0.59, 0.29, 0.0}
  \definecolor{darkbrown}{rgb}{0.4, 0.26, 0.13}
  \newcommand{\red}[1]{\textcolor{red}{#1}}
  \newcommand{\Red}[1]{\textcolor{red}{#1}}

  \newcommand{\cut}[1]{}%

\newcommand{\im}[1]{\ensuremath{#1}}

\renewcommand{\transition}{\xrightarrow{~~}}

\newcommand{\on}{Flw: \red{\checkmark}}

\newcommand{\els}{\mathit{else}\ }
\newcommand{\newfId}{\mathit{nfId}}
\newcommand{\und}{\mathit{\_}}
\newcommand{\ifs}{\mathit{if}\ }

\newcommand{\Cl}{\mathit{Cl}}    
\newcommand{\Att}{\mathit{Att}} 
\newcommand{\Mtds}{\mathit{Mtds}}  
        
\newcommand{\Prr}{\mathit{Pr}}  
\newcommand{\sll}{\mathit{sl}}    
\newcommand{\Cnt}{\mathit{Cnt}}     
\newcommand{\Lvl}{\mathit{Lvl}}      
\newcommand{\lev}{\mathit{lev}}        
                                                                                                                                                                   \newcommand{\PC}{\mathit{Pc}}                                                                                                                                                              \newcommand{\pcs}{\mathit{pcs}}                                                                                                                                                                     \newcommand{\pc}{\mathit{pc}}                                                                                                                                                                        \newcommand{\Qu}{\mathit{Qu}}  
\newcommand{\Ev}{\mathit{Ev}}  
\renewcommand{\mm}{\mathit{mm}}    
 
\newcommand{\invocc}{\mathit{invc}}
\newcommand{\fId}{\mathit{fId}}   
\newcommand{\dl}{\mathit{dl}}                                                                                                                                                                        \newcommand{\fut}{\mathit{fut}}                                                                                                                                                                       \newcommand{\Wrr}{\mathit{Wr}}  
\newcommand{\wId}{\mathit{wId}}                                                                                                                                                                    \newcommand{\newId}{\mathit{newId}}                                                                                                                                                              \newcommand{\new}{\mathit{new}} 
    
\newcommand{\eval}{\mathit{eval}} 
\newcommand{\level}{\mathit{level}}    
  
\newcommand{\comp}{\mathit{comp}}


%% file: title.tex
\title{ Security Wrappers for Information-Flow Control in Active Object Languages with Futures}
\author{Farzane Karami \inst{1} \and Olaf Owe \inst{1} \and Gerardo Schneider \inst{2}}
 \institute{ Dept.~of Informatics, University of Oslo \email{\{farzanka,olaf\}@ifi.uio.no}  \and Dept.~of Computer Science and Eng.,  Chalmers University of Technology \email{gerardo@cse.gu.se}}
\maketitle
\begin{abstract}
This paper introduces a run-time mechanism
for preventing leakage of secure information in distributed systems.
We consider a general concurrency language model,
where concurrent objects interact by asynchronous method calls and futures.
The aim is to prevent leakage of confidential information to low-level viewers. 
The approach is based on the notion of a \emph{security wrapper}, which encloses 
an object or a component and
controls its interactions with the environment.
A wrapper is a mechanism added by the run-time system
to provide protection of an insecure component according to some security policies.
The security policies of a wrapper are formalized based on 
a notion of security levels.
At run-time, future components will be wrapped upon need, 
while  only objects of \emph{unsafe classes} will be wrapped,
using 
static checking 
to limit the number of unsafe classes and thereby reducing run-time overhead. 
We define an operational semantics and prove that non-interference is
satisfied.
A service provider may use wrappers to protect its services 
in an insecure 
environment, 
and vice-versa: a system platform may use 
wrappers to protect itself 
from insecure service providers.
\keywords{Active objects \and Futures \and Information-flow security \and Non-interference \and Dynamic analysis \and Static analysis \and Distributed systems}
\end{abstract}

%% file: background.tex
%
%
\section{Introduction}
Given the large number of users and service providers involved in a distributed system, security is a critical concern.
It is thus
essential to analyze and control how confidential information propagates
between nodes.
When a program executes, it might leak secure information to
public outputs or send it to malicious nodes.
\emph{Information-flow control} approaches track
 how information propagates through the program during execution and prevent leakage of secure information~\cite{sabelfeld2003language}. 
 Program variables are tagged typically with security levels, for example, \textit{high} and \textit{low} to indicate secure and public data, and
the semantic notion of information-flow security is based on 
\textit{non-interference}~\cite{goguen1982security}.
This
means that
in any two executions,
when a program is run with a
variation of high inputs but the same low input values,
the low outputs will be the same (at least for locally deterministic programs).
This way, an attacker cannot see any variation between these two executions
since low outputs are independent of the high inputs. 
It is assumed that an attacker is a low level object that is communicating with the system.
 
\textit{Active object languages} are concurrent programming languages suitable for designing distributed and object-oriented systems.
The goal is to design an efficient, permissive, and precise
security mechanism that can be applied to these languages, supporting
concurrency and communication paradigms like \textit{futures}.
A future is a read-only component, which is created by a remote method call and
eventually will contain the corresponding return value~\cite{baker1977incremental}.
Therefore, the caller needs not
block while waiting to get the return value:
it can continue with other tasks and later get the value from the corresponding future.
Futures can be passed to other objects, called \emph{first-class futures}.
In this case, any object that has a reference to a future can access its content,
which may be a security threat if the future contains secure data.
Futures offer
a flexible way of communication and sharing results,
but handling them appropriately, in order to avoid security leakages, is challenging~\cite{KARAMI2019154}.

Our security mechanism is inspired by the notion of
\textit{wrappers} in~\cite{owe2009wrap}, where a wrapper encloses an object and enforces safety rules.
In this paper,
we suggest a notion of \emph{security wrapper}, which
wraps an object or a future
at run-time and performs security controls. 
Such wrappers are added to a high-level concurrent core language, supporting asynchronous method calls and futures.
Wrappers are \emph{invisible} in the sense that a wrapped  component  and its environment are not aware of it.
Security wrappers block illegal object communications, leading to leakage of confidential data to unauthorized object.
Moreover, they 
protect futures if they contain high return values, 
preventing illegal access by lower level objects.
The security policies and semantics of a wrapper are defined based on run-time security levels, which are resulted from a flow-sensitive and dynamic information-flow enforcement.
Therefore,
our dynamic approach guarantees some levels of permissiveness
and is precise since it deals with the exact run-time security levels.
The operational semantics of our security framework is provided in the style of Structured Operational Semantics (SOS), modeled in the Maude system~\cite{duran2007all}, giving an executable interpreter for programs written in our language. 
In order to have less run-time overhead, we suggest static
analysis to identify where security checking and wrappers are needed since
often only a few methods deal with secure information.
Assuming a sound static analysis, we prove that our proposed security mechanism
ensures the non-interference property in object communications.

In summary, our contributions are: i) a high-level core language as a model for service-oriented and distributed systems (Sec~\ref{language-syntax}), ii) a built-in notion of security wrappers for enforcing noninterference and security control in object interactions (Sec~\ref{our language}), for which we use static analysis to reduce the run-time
overhead (Sec~\ref{static}), and provide an executable operational semantics
for the dynamic information-flow enforcement and wrappers (Secs~\ref{op}, \ref{wrapper}), and iii) an outline of the proof that our model satisfies non-interference (the details are found in
the appendix Sec~\ref{appendix}).
\section{Background}
Information-flow control approaches detect illegal flows.
During program execution,
there are two kinds of leakage of information,
namely \textit{explicit} and \textit{implicit flows}~\cite{sabelfeld2003language}.
For simplicity, we assume two security levels $L$ (low) and $H$ (high).
An explicit flow happens by assigning a high variable to a low one ($l
:=h$), where the variables $l$ and $h$ store secret and public values.
One way to deal with this is to let the level of $l$ be changed to
$H$ (assuming $l$ is not directly observable by attackers).
In an implicit flow,
there is an indirect flow of information due to control flow structures. 
For example, in the \kw{if} statement: $l := 0; \ \kw{if} \ h \ \kw{then} \ l := 1 \  \kw{fi}$, the guard  $h$ is high, and it affects the value of $l$ indirectly.
In order to avoid implicit flows, a program-counter label ($\pc$) is introduced,
which captures the program context security level~\cite{sabelfeld2003language}.
If
the guard is high, then $\pc$ becomes high.

Information-flow control approaches are divided into two categories, static and dynamic~\cite{russo2010dynamic}.
Static analysis  
is conservative:
in order to be sound, it over-approximates security levels
of variables
(for example, it over-approximates a  formal parameter 
 to high, 
while at run-time a corresponding actual parameter
 can be low). 
This causes unnecessary rejections of programs,
especially when the complete program is not statically known,
as is usually the case in distributed systems.
On the other hand, static analysis has 
less run-time overhead,
since security checks are performed before program execution~\cite{russo2010dynamic}.
\emph{Dynamic information-flow} techniques perform security checks at run-time, and this introduces an overhead.
But they are more permissive and precise
since they deal with the exact security levels of variables instead of an over-approximation~\cite{hedin2012perspective}.

In what follows, we briefly explain some of the terminologies of information-flow security that we 
use 
in this paper:
\newline
\textbf{Security levels}. Variables are tagged with security levels, organized
by a partial order 
$\sqsubseteq$ and a join $\sqcup$ operator,
such that $L\sqsubseteq H$ and $L\sqcup H=H$.
Given two security levels, the $\sqcup$ operator 
returns the
least upper bound 
of them.
The policy rules describe the flow of information 
regarding 
security levels.
Inside a class,  declarations of fields, class parameters, 
and formal parameters 
may have  statically declared initial security levels.
These levels may change with statements.
At run-time, 
objects  are assigned 
security levels as well.
\\
\textbf{Flow-sensitivity}. 
In a \textit{flow-sensitive analysis}, variables start with their declared security levels (the ones without levels are assumed as $L$),
but levels may change after each statement.
In an assignment, the level of a left-hand-side
becomes high if $\pc$ is high, or there is a high variable on the right-hand-side.
The level of a left-hand-side becomes low
if the assignment does not occur in a high context,
and there are no high variables on the right-hand-side~\cite{russo2010dynamic}.
Otherwise, the security level of a variable does not change.
For example, a flow-sensitive analysis accepts the program 
$ h := 0; \ \kw{if}  \ h \ \kw{then} \ l := 1  \ \kw{fi};  \ \kw{return}  \ l;$
since $h$ is updated to low after the first assignment, and there is no leakage.
In an assignment $x:=e$, the security level of $x$ is updated to the join
of the security level of $\pc$ and expression $e$, which is the join of the security
levels of variables in $e$~\cite{russo2010dynamic}.
By a flow-sensitive dynamic analysis, security levels propagate to other variables, and precise levels are evaluated during execution. 
In~\cite{russo2010dynamic}, it is shown that although this technique is permissive, it is unsound 
because of implicit flows that happen in conditional structures with high guards.
The authors propose a conservative dataflow rule to avoid the implicit flows.
When the guard is high, the security levels of variables appearing on the left-hand side of
assignments in the taken and untaken branches are raised to high.
For instance, considering an initial environment $\Gamma=$\{$h\mapsto H, l_1\mapsto L, l_2\mapsto L$\} and the program: 
$\kw{if} \ h \ \kw{then}  \ l_1:=1; \ \kw{else} \ l_2:=0;  \ \kw{fi}$
when the condition is true, $\Gamma$ changes to $\Gamma=$\{$h\mapsto H, l_1\mapsto\ H, l_2\mapsto H$\}
for a sound flow-sensitive analysis~\cite{russo2010dynamic}.

Active object languages
have gained a lot of attention in recent years~\cite{boer2017survey}.
They are based on a combination of the \emph{actor model}~\cite{agha1985actors} and object-oriented features.
Some well-known active object
languages are
Rebeca~\cite{sirjani2001compositional,sirjani2004modeling},
Scala/Akka~\cite{haller2009scala,wyatt2013akka}, Creol~\cite{johnsen2007asynchronous}, ABS~\cite{johnsen2011abs}, Encore~\cite{brandauer2015parallel}, and ASP/ProActive~\cite{caromel2005theory,caromel2006proactive}.
In the  call/return paradigm without futures,
a method call and the corresponding return value are sent by message passing between the
caller and callee objects.
An object has an external queue for receiving these messages.
In a naive model, a caller waits while the callee computes the
call, such a call is blocking, which is undesirable.

The notion of futures is a common mechanism 
in active object languages
for avoiding blocking calls~\cite{boer2017survey,KARAMI2019154}.
When 
a remote method call is made, a future object with a unique identity is created.
The caller may continue with other processes while the callee is
computing the return value.
The callee sends back the return value to the corresponding future object, and then
the future is called \emph{resolved}.
The future mechanism gives the flexibility that in a service-oriented system, a client program
waits only when it needs to get the value from a future object.
Futures can be explicit with a specific type and access operations like in ABS
or can be implicit with automatic creation and access~\cite{boer2017survey}.
An explicit future in ABS is created as $f := o!m(\many{e}); v:=f.\kw{get}$,
where $f$ is a future variable of type $\mathit{Fut}[T]$,
and $T$ is the type of the future value.
The symbol ''!''
indicates an asynchronous method call $m$ of object $o$,
and the future value is retrieved with a \kw{get} construct.
In contrast, a synchronous call is denoted by $o.m(\many{e})$, which blocks the caller until the return value is retrieved.

\textbf{Information-flow security with futures} \\
Static {analysis} 
would be
difficult for programs with futures, where
the result of a call is no longer syntactically connected to the call, compared to the call/return paradigm in
languages without futures~\cite{KARAMI2019154}.
For example, a future may be created in one module and received as a parameter in another.
A future
may not statically correspond to a unique call statement.
One can overestimate the set of possible
call statements that correspond to this given future parameter, but it needs 
access to the whole program, which is a problem with distributed systems.
Moreover, the return values of these
overestimated calls may have different security levels.
In this case, the highest security level for a future is overestimated.

In static analysis there is limited static knowledge about object identities.
Therefore a language  may compensate by including  a syntactic construct
for testing security levels, as in~\cite{oweramezani17dpm}.
By testing relevant security levels, one may achieve fine-grained security control
(at the cost of added branching  structure).
A static analysis that assumes future references as low allows passing of futures.
However, the exact security level of a future value is revealed when it becomes resolved.
A dynamic approach is required to control access to a future value at run-time when it is resolved, and if the value is high it needs protection.
The concept of futures makes static checking less precise, and the need for complementary run-time checking is greater, as provided in the present paper.
\subsection{A  core language syntax}\label{language-syntax}
In order to exemplify our security approach,
the security semantics (in Sec.~\ref{op}) is embedded in a core language,
using a  syntax similar to Creol/ABS. 
Our security approach can be applied to any active object language with futures.
In our core language, all remote object interactions are made by means of futures.
Therefore, the result of a remote call always is returned back to the
corresponding future.
Figure~\ref{fig:unifiedsyntax} provides 
the core language syntax.
The statement 
$o!m(\many{e})$ is an asynchronous call toward object $o$ without waiting for the result and associating a future for the call.
Here $o$ may be a list of objects, broadcasting to those objects, and
$\many{e}$ is a list of actual parameters. 
The statement $q!o.m(\many{e})$ is an asynchronous and remote method call, which
creates an explicit
future and assigns the identity to the {future variable} $q$. 
A return value can be accessed with a get statement 
$ q?(x)$,
which blocks while waiting for the corresponding future to become resolved
and then assigns the value to the variable $x$.
The statement $q?(x)$ is similar to $x=q.\kw{get}$ in ABS.
%

 \begin{figure}[t]
\centering
\begin{footnotesize}
	\begin{tabular}{ll}
		\emph{Basic constructs}\\ \hline 
		${x:=\kw{new}_{\mathit{lev}}\ c(\many{e})}$ &{object creation with the security level $\mathit{lev}$}\\
		${\kw{return}\ e}$ &{creating a method result/future value}\\
		${\kw{if}\ b\ \kw{th}\ s\ [\kw{el}\ s']\ \kw{fi}}$
		&{if statement ($b$ a Boolean condition)}
		\\
		${!o.m(\many{e})}$ &simple asynchronous remote call
		\\
		${q!o.m(\many{e})}$ &remote asynchronous call, future variable $q$\\
		$ q?(x) $  & {blocking access operation on future $q$}
	\end{tabular}
\end{footnotesize}
	\caption{Core Syntax.}
	\label{fig:unifiedsyntax}
\end{figure}
Figure~\ref{figabs} illustrates an example of a health care service in our core language, involving futures for communication and sharing of secure medical records. 
High variables are emphasized based on static analysis/user specifications,
in this case reflecting patients' medical test results.
The problem the system is trying to solve is to prevent personnel and patients
with lower-level access from accessing the medical records.
In this example,
the server, specified by class \textit{Service}, searches for the test result of a patient
and publishes it to
the patient and associated personnel through the  \textit{proxy} object.
The server uses futures to communicate test results to the \textit{proxy},
thus it does not wait for the results and is free to respond to any client request.
Instead, the \textit{proxy} waits for the results from the laboratory and publishes them. 
In line 10, a produce cycle is initiated between the \textit{server} and \textit{proxy}.
In line 13, the server searches for the test results of this patient by sending a remote asynchronous call
to the laboratory \textit{q!lab.detectResult}, 
where \textit{q} is the future variable and 
\textit{lab} is the callee.
In line 14, the future 
is passed to object \textit{proxy} for publishing, and
this object waits and assigns the result to the variable \textit{r} as in line 18.
Then object $\mathit{proxy}$ sends $r$ to the patient and health personnel.
\begin{figure}[t]
\begin{footnotesize}
\begin{abs}
data type Result = ...         // definition of medical data
interface ServiceI { Void produce() ... }	
interface ProxyI { Void publish(Fut[Result] q, PatientI a, List[PersonnelI] d) ... }	
interface LabI { $\textsf{Result}_{\Red{\textsf{H}}}$ detectResult(PatientI a) ... }	
interface PatientI { Void signal($\textsf{Result}_{\Red{\textsf{H}}}$ r) ... }	
interface PersonnelI { Void signal($\textsf{Result}_{\Red{\textsf{H}}}$ r) ... }	
interface DataBase { ... }		
class Service(LabI lab, DataBase db) implements ServiceI {
 ProxyI proxy = new Proxy(this); 
 !this.produce(); $ \ \ $ // initial action, starting a produce cycle		
 Void produce() { Fut[Result] q;$ \ $  PatientI a; $ \ $ List[PersonnelI] d = Nil;  
   ....  $ \ $ $ \ \ $ // finding a patient and the associated personnel in database
   q!lab.detectResult(a);
   !proxy.publish(q, a, d); } } $  \ $ // sending the future (q), no waiting 
		
class Proxy(ServiceI s) implements ProxyI{ $\textsf{Result}_{\Red{\textsf{H}}}$ r;
 Void publish(Fut[Result] q, PatientI a, List[PersonnelI] d) { 
 q?(r); $ \ $ // waiting for future and assigning the value to r 
 !a.signal(r); $ \ \ $     //     $\Red{ r\ is\ \mathit{now} \ H}$
 !d.signal(r); $ \  $  // multicasting, $\Red{ r\ is\ H}$
 !s.produce(); } }
\end{abs}
\end{footnotesize}
	\caption{
	Example of sharing \Red{\emph{high}} 
	patients' test results by means of futures
}
	\label{figabs}
\end{figure}
A static analysis
over-approximates the security levels of test results as high, which
leads to
rejections of information passing.
Note that the two signal calls in class $\mathit{Proxy}$ 
would not be allowed if we only use static checking
since we cannot tell which patients and personnel have a
high enough level.
A static analysis which considers references as low
allows passing a future $q$ to object \textit{proxy} (line 14),
but later when it is resolved, the future value can be high, and $\mathit{proxy}$ compromises security by sending this value to other objects.
\section{A 
framework for 
non-interference%
} \label{our language}
Like Creol/ABS, our core language is equipped with
behavioral interfaces,
which means that
created objects are typed by interfaces, not classes~\cite{johnsen2007asynchronous}.
As a result, remote
access to fields or methods that are not declared through an interface is impossible.
Therefore, observable behavior of an object is limited to its interactions through remote method calls.
Illegal object interactions are the ones leading to an information-flow from high to low level
information holders.
In this paper, we exploit the notion of wrappers
to perform dynamic checking for enforcing non-interference in object interactions. 
The operational semantics of wrappers is explained in Sec.~\ref{wrapper}.
An object can reveal confidential information through outgoing method calls
by sending actual parameters with high security levels to lower level objects.
In this case, a wrapper based on the security policies, blocks illegal communications.
A future is wrapped when there is a high return value, the
wrapper controls access to the value and blocks illegal ones.
Security policies of a wrapper are based on run-time security levels.
Inside an object, in order to compute the exact security levels of created messages or return values,
flow-sensitivity must be active,
using dynamic information-flow enforcement.
The operational semantics of our dynamic flow-sensitive enforcement is given in Sec.~\ref{op}. 

We can be conservative and wrap all objects and correspondingly activate 
flow-sensitivity,
but this will cost run-time overhead.
In order to be more efficient at run-time,
it is important to perform dynamic checking only for components where it is necessary.
We benefit from a class-wise static analysis to categorize a class definition as \textit{safe} or \textit{unsafe}.
\emph{A class is safe if it does not have any method calls with high actual parameters and return values.
A class is unsafe if it has a method call with at least one high actual parameter or a high return value.}
Objects created from unsafe classes are wrapped with active flow-sensitivity.
While objects from safe classes do not need a wrapper or active flow-sensitivity.
This will make the execution of objects of safe classes faster, as we avoid a potentially large number of run-time checks and wrappers.

By combining our dynamic approach with static analysis, we benefit
from the advantages of both approaches, shown in the table below,
where $+$ is better that $-$.
Namely, we do not do runtime checking of statically safe classes.
The gray cells demonstrate our approach's advantages.
{\small
\begin{center}
	\begin{tabular}{ c |c |c }
		 & \textbf{static analysis} & \textbf{dynamic analysis}  \\ \hline
		runtime efficiency & \cellcolor{shadecolor} +  & - \\  \hline 
		exact security and permissiveness
 & - & \cellcolor{shadecolor} +  
	\end{tabular}
\end{center}}
\subsection{Static analysis} \label{static}
Our security approach can be combined with a sound static over-approximation for detecting security errors and safe classes, for instance, the one proposed in~\cite{oweramezani17dpm},
which is more permissive (to classify a class as safe) than the static analysis indicated here, in that  
high communication is considered secure
as long as the declared levels of parameters are respected.
Class fields and methods' formal parameters and return values are declared with maximum security levels (maximum of possible level that can be assigned at run-time).
Local variables that do not have a declared security level, they start with security level $L$ (as default)
but
may change after each statement due to flow-sensitivity.
Since an object encapsulates local data and fields, there is no need
to control the flow of information inside an object.
However, since all object interactions are done by method calls,
typing rules are defined to control security levels of return values and actual parameters~\cite{oweramezani17dpm}.
A class is defined as safe if the confidentiality of each method is satisfied.
The confidentiality of a method is satisfied if
the typing rules for its return value and actual parameters are satisfied. 
The typing rules check that
each occurrence of an actual parameter and a return value are not high; then, the class is safe; otherwise, it is unsafe and needs dynamic checking..

We categorize safe and unsafe classes for the example in Fig.~\ref{figabs}.
The interface laboratory \textit{LabI} has a method with a high return value (\textit{detectResult}).
Thus the object \textit{lab} is unsafe and flow-sensitivity is active to
compute the security level of the return value at run-time.
The class \textit{Proxy} is unsafe since it has at least one method call with a
high actual parameter,
thus object \textit{proxy} is active flow-sensitive and wrapped.
Based on the security levels of a patient and personnel,
the wrapper of object \textit{Proxy}
decides whether to send out the \textit{signal} messages or block them.
Precise level information for objects is available at run-time.
\subsection{
Security semantics} \label{op}
We here discuss the operational semantics of our core language with the embedded
notions of flow-sensitivity and security wrappers.
The small-step operational semantics is defined by a set of rewrite
logic (RL) rules in the Maude format.
Maude 
uses sorts to define data types and data structures
together with operator functions to manipulate values and
define the value set of a sort.
Operator functions are defined by equations.
Rules and equations can be conditional
as $lhs \longrightarrow rhs \  \mathit{ if} \mathit{condition}$ for rules
and $\mathit{ rhs} = \mathit{ lhs} \ \mathit{ if \ condition}$ for equations.
A systems is modeled as  multisets
consisting of 
the relevant components and messages, and the outcome can be different based on non-deterministic applications of rules.
Maude has a built-in understanding of
associativity, commutativity, and
identity, which are used
to define multisets.
A system state is modeled as a \textit{configuration}, which is
a multiset
of objects (with or without  active flow-sensitivity), queues,
messages, futures, and wrappers.

\begin{footnotesize}
\begin{maude}
sorts Class Object Queue Msg Future Wrapper Configuration .
subsorts Class Object Queue Msg Future Wrapper < Configuration .			
op none : $ \longrightarrow $ Configuration .
op _ _ : Configuration Configuration $\longrightarrow$ Configuration [assoc comm id: none] .     
\end{maude}
\end{footnotesize}

The \textit{Configuration} sort is a super sort of the components' sorts. (Classes are included in a configuration to provide static information
about fields and methods.)
The first operator 
(which does not take any arguments) defines $none$ as a constant of sort \textit{Configuration}, representing an empty state. 
The binary multiset operator 
defines the \textit{Configuration} sort as a multiset.
The underscores (\_) determine where the arguments should be placed.
In Maude, an object is represented by a term
$<o:c \ | \  \mathit{Att}_{1} : v_{1}, ..., \mathit{Att}_{n} : v_{n}>$, where $o$ is the object name, $c$ is its class name,
the $\mathit{Att}_{i}$'s are the object attributes, and the $v_{i}$'s are the corresponding meta variables, representing instances~\cite{olveczky2018designing}.

Figure~\ref{fig:object-model} represents the components of a configuration.
An \textbf{object} component 
is represented as:
$ <o:c\  |  \ \Att:a, \Prr:(l,\mathit{sl}), \Cnt:n, \Lvl: lev>$,
where $o$ is the object identity, $c$ is the class name that the object is created from, $a$ is
the object fields' state (a mapping from instance variables to values), 
the pair $(l,\mathit{sl})$ represents the current active process,
where $l$ is the
state of local variables (a mapping from local variables to values),
and $\mathit{sl}$ is a list of statements in the active process.
The attribute $n$ is a natural number
to assign unique identities to each method invocation created by the object.
Variable $\mathit{lev}$ is the security level of the object ($\lev \in \{L,\ H\}$),
which is assigned by a programmer at the time of creating an object, using the syntax $x := \kw{new}_{\lev} \ c'(..)$.
A \textbf{flowsen-obj} component shows an object with active flow-sensitivity, with extra fields
$\on$,
representing active flow-sensitivity,
and $\pcs$ denotes a stack of context security levels inside an object.
An empty stack ($\pcs=\mathit{emp}$) corresponds to $\pc=L$, representing a low context security level.
A non-empty stack ($\pcs\neq \mathit{emp}$) corresponds to $\pc=H$, representing a high context security level.
A \textbf{class} is represented as:
$ <\Cl:c \ | \ \Att:a', \Mtds:\mm, \Cnt:f>$,
where $c$ is the class name,
$a'$ 
is the class fields state (attributes),
$\mm$ is a multiset of method declarations (with local variables and code), 
and $f$ is a natural number for generating unique object identities.
The \textbf{unsafe-class} component has an additional field $\on$,
denoting that all objects created from this class will be wrapped with active flow-sensitivity.
{\small \begin{figure}[t]
		\centering
		$$\begin{array}{lll}
		\OPRULE{object} & <o:c \ | \ \Att:a, \mathit{Pr}:(l,\sll), \Cnt:n, \Lvl: \lev> 	\\
		\OPRULE{flowsen-obj} & <o:c \ | \ \Att:a, \mathit{Pr}:(l,\sll),  \Cnt:n, \Lvl: \lev, \on, \PC: \pcs  >	
		\\
		\OPRULE{class} & <\Cl:c \ | \ \Att:a', \Mtds:\mm, \Cnt:f>\\[0.3em]
		\OPRULE{unsafe-class} & <\Cl:c \ | \ \Att:a', \Mtds:\mm, \Cnt:f, \on >\\
		\OPRULE{queue} & <\Qu:o \ | \ \Ev : \mathit{msg} > \\
		\OPRULE{invoc-msg}  & \invocc(\fId,m,\dl) \  o \ \mathit{to} \ o' \\
		\OPRULE{comp-msg}  & \comp(d) \ o \ \mathit{to} \ \fId \\
		\OPRULE{future} & \fut(\fId, d) \\
		\OPRULE{wrapper}  & \{\Wrr : \wId , \Lvl :\lev \ | \ \textit{config}\}
		\end{array}$$
		\caption{The components of a configuration.}
		\label{fig:object-model}
\end{figure} }

The \textbf{queue} component represents the external queue of an object for storing method invocations toward the object.
The queue is associated with 
object $o$, 
and $\mathit{msg}$ is a multiset
of stored messages. 
The \textbf{invoc-msg} component represents an invocation message,
where $\fId$ is the future identity, $m$ is the called method name,
and $\dl$ is a list of actual parameters.
The \textbf{comp-msg} component represents a completion message,
where $d$ is the return value from object $o$ to the future $\fId$.
The 
\textbf{future} component shows a resolved future 
with 
identity $\fId$ and value $d$, and $\fut(\fId, \und)$ indicates
an unresolved 
future with no value.
A \textbf{security wrapper} is represented as:
$\{ \Wrr:\wId , \Lvl :\lev  | \textit{config}\}$,
where $\wId$ is the wrapper's identity, $\lev$ the level,
and 
the meta variable $\mathit{config}$ denotes the configuration inside the wrapper.
A wrapper wraps a future or an object and correspondingly,
gets the identity and level of that component.
A wrapper acquires and records the security levels of
destination objects before sending messages, which we ignore here.

Operational semantics is defined by a number of rewrite rules,
which are applied non-deterministically. 
If the left-hand-side of a rule matches a subset of a configuration,
this rule can be applied, and the
left-hand-side pattern is replaced by the corresponding
instance of the right-hand-side.
Since configurations can be nested inside wrappers,
the rules can also be applied to inner configurations. 
Non-overlapping rules may be applied in parallel
indicating non-determinism and concurrency.
Each rule involves one object, implying that objects are executing independent. 
For simplicity, in the right-hand-side of the rules, fields that are as in
the left-hand-side are ignored and indicated by~$...$,
but 
changed fields are shown.

Figure~\ref{operational-semantics-flowsen} represents 
the 
flow-sensitivity semantics of objects.
The \textbf{new} rule shows the command $x:= \kw{new}_{\lev} c'(\many{e})$ in the active process of object $o$, creating a new object from an unsafe class $c'$.
It creates an active flow-sensitive object inside a wrapper.
The wrapper gets the same identity and level as the new object, $\Wrr: \newId, \Lvl: \lev$,
and it wraps the queue as well.
$\newId$ is an abbreviation for a function $\new(c',f)$,
which creates a unique object identity from
the class name $c'$ and object counter $f$.
This rule assigns $\newId$ to the variable $x$
and increases the object counter by one in the class $c'$.
The active process of the new object is initialized with ($\mathit{empty,idle}$),
denoting an empty state of local variables
and an empty statement list (no active process), respectively.
The new object level is $\lev$ as it is specified in the command $ \kw{new}_{\lev} c'(\many{e})$, if not, 
the level is assumed low.
The other fields are initialized as $\on$ and $\mathit{Pc:emp}$, where the latter represents an empty stack of context security level.
The queue is associated to the new object's identity
with no stored messages ($\mathit{noMsg}$).
The semantics of the actual class parameters
is treated like parameters of an asynchronous call
$!x.\mathit{init}(\many{e})$,
where \emph{init} is the name of the initialization method
of a class.
{\small  \begin{figure}[!th]
		$$\begin{array}{lll}
%
		\OPRULE{new} 
		&& < o : c \ | \Att: a,  \Prr: (l,  x := \kw{new}_{\lev} \ c'(\many{e});  \sll),  \Cnt: n, \Lvl: \mathit{lev'} > \\
		&&< \Cl:c'  | \Att: a' , \Mtds: \mm , \Cnt:f , \on>
		\\ &\transition &
		< o : c  |..., \Prr: (l, x := \newId; !x.\mathit{init}(\many{e}); \sll), ...\! \!>\  < Cl:c'  | ... , \Cnt: f+1, ...>\\
		&&\Blue \{  \Wrr : \newId,  \Lvl : \lev  | < \newId : c'  | \Att: a', \mathit{this} \mapsto \newId ,\Prr:(\mathit{empty},\mathit{idle}), \\
		&& \Cnt:1, \Lvl: \lev, \on, \PC:\mathit{emp} > \
		< \Qu:\newId \ | \ \Ev: \mathit{noMsg} >  \Blue \}
		\\ [0.3em]
		
%
		\OPRULE{assign} 	
		&&< o : c \ | \ \Att: a, \Prr:(l, x:= e; \sll), \Cnt:n,  \Lvl:\lev,  \on,  \PC:\pcs >
		\\ &\transition &
		\ifs (\mathit{dom(x,l)}) \ 
			< o : c \ | \ ..., \Prr: (\mathit{insert}(x, [e]_{\mathit{level}([e])\sqcup \pc}, l), \sll),  ... > \\
		&&\els < o : c \ | \ \Att: \mathit{insert}(x, [e]_{\mathit{level}([e])\sqcup \pc}, a), \Prr: (l, \sll), ...> 
		\\ [0.3em]
		
		\OPRULE{if-low}	
		&&< o : c \ | \ \Att: a, \Prr: (l,\kw{if} \ e \ \kw{th}\ \sll' \ \kw{el}\ \sll'' \ \kw{fi} ;\sll), ..., \on,\PC: \pcs> 
		\\& \transition&
		\mathit{if} ([e] = T) 
		 < o : c  \ |  ...,\Prr: (l, \sll' ;\sll), ...\! \! >  \mathit{else} < o : c   |  ..., \Prr: (l, \sll'' ; \sll),...\!> \\
		&\mathit{if}& \level ([e]) = L \ \wedge \ \pcs= \mathit{emp}
		\\ [0.3em]
		
		\OPRULE{if-high}	
		&&< o : c \ | \ \Att: a, \Prr: (l, \kw{if} \ e \ \kw{th} \ \sll' \ \kw{el} \ \sll'' \ \kw{fi} ;\sll), ...,\on,\PC: \pcs > 
		\\& \transition&
		\mathit{if}\ ([e] = T) \
		 < o : c | ..., \Prr: (l, \sll' ; \mathit{endif}(\sll'') ;\sll),... ,\PC: \pcs.\mathit{push(H) }> \\
		&& \els < o : c \ | \ ...,\Prr: (l, \sll'' ; \mathit{endif}(\sll') ; \sll), ..., \PC: \pcs.\mathit{push(H)} > \\
		&\ifs  &\level[e] =H \ \vee \ \pcs\neq \mathit{emp}
		\\ [0.3em]
		
		\OPRULE{endif} 
		&&< o : c |  \Att: a, \Prr: (l, \mathit{endif}(\sll') ;\sll), \Cnt: n,\Lvl: \lev,\on, \PC: \pcs\ > 
		\\& \transition&
		< o : c  | ...,\Prr: (l,\mathit{\mathit{update}_{H}}(\sll');\sll),..., \PC: \pcs.\mathit{pop()} >
		\\ [0.3em]
		
		\OPRULE{start}	
		&& < o : c \ | \ \Att:  a,  \Prr:  (\mathit{empty},\mathit{idle}), \Cnt: n, \Lvl: \lev, ... > \\
		&& < \Qu:o \ | \ \Ev: \mathit{msg} \ \invocc( \mathit{Id}, m ,  \dl)>  
		\\& \transition&
		< o : c  |  \Att:  a, \Prr: ([ \mathit{label}\mapsto \mathit{Id}, \many{x} \mapsto \dl, \many{l}\mapsto \many{l_{0}}],  \sll), ...\! \!>  < \Qu:o  |  \Ev: \mathit{msg} >\\
		&& \kw{where} \ \text{method} \ m \ \text{binds} \ \text{to} \ m(\many{x})\{ \many{l_0}; \ \sll;\}	\ \text{with} \ \text{initial} \ \text{local state
		} \ \many{l_{0}}
		\\ [0.3em]
		
		\OPRULE{call}	
		&& < o : c \ | \ \Att:  a,  \Prr:  (l, q ! e. m(\many{e}); \sll),  \Cnt: n, \Lvl: \lev,  ...> 
		\\& \transition&
		< o : c  \ | \ \Att:  a, \Prr: (l,  q := \newfId;  ! e . m(\many{e});   \sll), ...> \
		\fut(\newfId , \und )
		\\ [0.3em]
		
		\OPRULE{call'} 	
		&& < o : c \ | \ \Att: a,  \Prr: (l,  ! e . m(\many{e});  \sll),   \Cnt: n, \Lvl: \lev,  \on,\PC: \pcs >
		\\& \transition&
		< o  : c  | ... , \Prr: (l,\sll),  \Cnt: n + 1 , ...>  (\invocc( \newfId, m ,  [\many{e}])  \ o \ \mathit{to} \ [e])_{ \level([\many{e}]) \sqcup \pc} 
		\\ [0.3em]
		
		\OPRULE{return} 	
		&&< o : c |  \Att: a, \Prr: (l, \kw{return}(d); ), \Cnt: n, \Lvl: \lev,  \on, \PC: \pcs > 
		\\& \transition&
		< o : c  |  ...,\Prr:(\mathit{empty}, \mathit{idle}), ... > (\comp(d) \ o \ \mathit{to} \ \eval(\mathit{label}, l))_{ \level(d) \sqcup \pc} 
	\end{array}$$
	\caption{Flow-sensitivity operational semantics, where $\newId = \mathit{new}(c',f)$, $[e]= \mathit{eval}(e,(a \# l))$, $T=\mathit{true}$, and $\mathit{nfId} = \mathit{new}(o,n)$.}
	\label{operational-semantics-flowsen}
	\end{figure}}

In an active flow-sensitive object,
the evaluation result of an expression $e$ is denoted by
$[e]$, which
returns back 
both the value ($d$) and security level ($\mathit{lev}$) as $d_{\lev}$.
If this value is assigned to a program variable $x$,
the binding $x\mapsto d_{\lev}$ is added to
the corresponding state in the object.
At the time of analysis, a variable without a security level is assumed low (which might change afterward).
The function $\mathit{level}([e])$ returns back the security level of $[e]$ ($\lev$). If $\many{e}$ is a list of expressions, then $\mathit{level}([\many{e}])$ returns the join of security levels of expressions in $\many{e}$.
In addition, $[e]$ is an abbreviation for 
$\mathit{eval}(e,a\#l)$.
The $\mathit{eval}$ function evaluates an expression $e$ by considering 
the map composition of $a$ and $l$ (i.e., $a\#l$),
denoting that the binding of a variable name in the inner scope $l$ shadows
any binding of that name in the outer scope $a$~\cite{KARAMI2019154}.
Considering a variable $x$, the function $\mathit{eval}(x,a\#l)$ returns $\mathit{eval}(x,l)$ if $l$ has a binding for $x$; otherwise, $\mathit{eval}(x,a)$.

The Rule \textbf{assign} shows an assignment statement of $(x := e)$
in an active flow-sensitive object.
The function $\mathit{dom}(x,l)$ is true if variable $x$ is in the local state $l$,
then $x$ with its new value and level is updated in $l$
by a function $\mathit{insert}$;
otherwise, it is updated in the attribute state $a$.
Variable $x$ is updated to $x\mapsto [e]_{\mathit{level}([e])\sqcup pc}$ in the corresponding domain ($\pc=H$ if $\pcs\neq \mathit{emp}$, and $\pc=L$ if $\pcs= \mathit{emp}$).
A variable is added to the domain with at least the $\pc$ level.
This rule inserts variables containing the identity of a new object or future with the $\pc$ level,
although they are not confidential.
The \textbf{if-low} is applied when the guard's security level is low and $\pcs=\mathit{emp}$, here $T$ denotes $\mathit{true}$.
The guard is evaluated and the corresponding branch is taken.
When the guard's security level is high or the $\pcs$ stack is not empty, similar to the approach in ~\cite{russo2010dynamic},
the security levels of variables appearing in both branches are raised to high to avoid implicit flows.
This is captured in the \textbf{if-high} rule, where
the guard's security level ($\mathit{level}([e])=H$) is pushed to the $\pcs$ stack, resulting in $\pc=H$.
In addition, an auxiliary function $\mathit{endif(sl')}$ is defined to
mark the join point of the \kw{if} structure, where $\sll'$ is the untaken branch.
In the \textbf{endif} rule, the function $\mathit{update}_{H}(\sll')$ raises the security levels of variables appearing in the left-hand-side of assignments in $\sll'$
to high.
Moreover,
the last element of $\pcs$ is removed ($\pcs.\mathit{pop()}$), reflecting the previous context level.

The \textbf{start} rule is applied when there is no active process inside an object
(the object is \textit{idle}), and there is an invocation message in the queue.
The object starts executing the corresponding method.
The rule captures the method's body $\sll$ as the active process, and the object's local state is updated, binding
formal parameters to the actual ones and storing the future identity ($\mathit{Id}$) as $\mathit{label}$.
The $\mathit{label}$ variable is used later
to return back the result of a method to the corresponding future.
In the rules, we do not cover local calls, which do not involve
object interactions (therefore, less interesting here).
The \textbf{call} rule deals with 
an asynchronous and remote method call $q ! e . m(\many{e})$,
where $q$ is a 
future variable, and $e$ is the callee.
The call is reduced to $! e . m(\many{e})$, and a (not resolved) future with identity
$\newfId$ is created.
The ${\newfId}$ is an abbreviation for function $\mathit{new}(o,n)$, which creates a unique identity from
the object name $o$ and $n$. 
The new identity ${\newfId}$ is assigned to the future variable $q$.
The rule \textbf{call'} shows an asynchronous call $! e.m(\many{e})$ in an active flow-sensitive object, which creates
an invocation message toward the callee.
The invocation message contains $\newfId$, method name $m$, and actual parameters $\many{e}$.
The security level of the invocation message is $ \mathit{level}([\many{e}]) \sqcup pc$, which
is the join of security levels of the actual parameters in $\many{e}$ and $\pc$.
The level $ \mathit{level}([\many{e}]) \sqcup pc=L$ if all the actual parameters in $\many{e}$ are low and $\pc=L$.
The \textbf{return} rule interprets
a \kw{return} statement in an active flow-sensitive object,
which creates a completion message toward the corresponding future ($\mathit{eval(label}, l)$),
and the object becomes idle.
The security level of the completion message is $\mathit{level}(d) \sqcup pc$, and $d$ is the return value.
We assume that each method body ends with a \kw{return} statement.
\subsection{Operational semantics of security wrappers}\label{wrapper}
A wrapper 
is assigned to an object or a future by the run-time system.
All invocation messages generated from a wrapped object first meet the wrapper for security checking.
The \textbf{w-invc} rule in Fig.~\ref{operational-semantics-wrapper}, represents a wrapper with an invocation message inside, and
if the security level of the message
is less than or equal to the destination object level 
($\mathit{level}(o')$),
then the wrapper allows the message to go out.
Otherwise,
the wrapper eats the message, and it disappears from inside.
In this case, an object performing a get will deadlock because the invocation message was deleted.
We have extended the approach with a notion of errors,
so that the deletion of an invocation message results in an error value in the
corresponding future component.
This can be combined with 
an exception handling mechanism 
such that an exception is raised when a get operation tries to  access 
an error value.
However, as this is 
beyond the scope of this paper,
we ignore the exception handling part.
We simply indicate exceptions by assignments with \kw{error} in the right-hand-side.
The \textbf{fut-get'} rule represents the case when a future value is \kw{error},
and an object performs a get command $\fId ?(x)$ raises an exception instead of being blocked.
{\small \begin{figure}[t]
		$$\begin{array}{lll}
%
		\OPRULE{w-invc} 
		&& \Blue \{ \ \Wrr  :  o,  \Lvl : \lev \ | \ 
		(\invocc( \fId,  m ,  \dl 
		) \ o \ \mathit{to} \ o' )_{\mathit{lev'}}  \mathit{config} \ \Blue \} \ \fut(\fId , \und )
		\\& \transition&
		\mathit{if} (\lev'  \sqsubseteq  \level(o')) \
		 \Blue \{  \Wrr  :  o,  \Lvl : \lev  |  \mathit{config}  \Blue \} (\invocc(\mathit{fId}, m ,  \dl ) \ o \ \mathit{to} \ o')_{\lev'} \\ 
		 && \hspace{3mm}\fut(\mathit{fId} , \und ) \\
		&& \els \Blue \{ \ \Wrr  :  o, \ \Lvl : \lev \ |  \   \mathit{config} \ \Blue \} \ \	\fut(\mathit{fId} , \kw{error} )
		\\ [0.3em]
		
		\OPRULE{fut-get'} 
		&&\fut(\fId, \kw{error})  \ < o : c \ |  \ \Att: a, \Prr: (l,  \fId ?(x);  \sll),  \Cnt: f,  \Lvl: \lev,...> 
		\\ &\transition&
		\fut(\fId, \kw{error}) \ < o : c \ | \  ..., \Prr: (l,  x := \kw{error};  \sll), ... > 
		\\ [0.3em]
		
		\OPRULE{invc-w'} 
		&&\Blue \{ \ \Wrr  :  o, \ \Lvl : \lev \ | \ \mathit{config} \ \Blue \} \ \ \invocc( \mathit{Id}, \ m , \ \dl) \ o' \ \mathit{to} \ o
		\\ &\transition&
		\ifs  (\forall i : \level(\dl_{i})  \sqsubseteq  \varLambda [m,i] ) \
		 \Blue \{  \Wrr  :  o, \Lvl : \lev  |  (\invocc( \mathit{Id},  m ,  \dl )  \ o' \ \mathit{to} \ o)  \mathit{config}    \Blue \}  \\
		&&\els \Blue \{ \ \Wrr  :  o, \ \Lvl : \lev \ |  \ \mathit{config} \ \Blue \} 
		\\ [0.3em]
		
		\OPRULE{invc-q} 
		&&< \Qu:o | \Ev: \mathit{msg} > \invocc( \mathit{Id}, m , \dl )  \ o' \ \mathit{to} \ o = 
		< ... | \Ev: \mathit{msg} \ \invocc( \mathit{Id}, m , \dl)> 
		\\ [0.3em]
		
		\OPRULE{w-fut}
		&&\fut(\fId, \und) \ \	(\comp(d) \  o' \ \mathit{to} \ \fId)_{\lev} 
		\\& \transition&
		\ifs (\lev  = H) \ 
		 \Blue \{ \ \Wrr  :  \fId,  \Lvl : \lev \ |  \ \fut(\fId, d)\ \Blue \} \  \els \fut(\fId, d)
		\\ [0.3em]
		
		\OPRULE{w-get} 
		&&\Blue \{ \ \Wrr  :  \fId, \ \Lvl : \lev' \ | \ \fut(\fId, d) \ \Blue \} \\ 
		&&< o : c \ | \  \Att: a, \ \Prr: (l, \ \fId ?(x); \ \sll),  \ \Cnt: f, \ \Lvl: \lev,\ ...>
		\\& \transition&
		\mathit{if}  (\lev \sqsubseteq \lev')\   \Blue \{  \Wrr  :  \fId, \Lvl : \lev'  | ...  \Blue \} < o : c \ | \  ...,  \Prr: (l,  x := \kw{error};  \sll),  ... > \\
		&&\els \Blue \{ \ \Wrr  :  \fId, \Lvl : \lev'  | \ ... \ \Blue \} \;  < o : c \ | \  ..., \Prr: (l,  x:= d;  \sll),  ... >
		\\ [0.3em]
		
		\OPRULE{fut-get} 
		&&\fut(\fId, d) \  < o : c \ | \  \Att: a, \Prr: (l,  \fId ?(x);  \sll),  \Cnt: f, \Lvl: \lev, ...>
		\\& \transition&
		\fut(\fId, d) \  < o : c \ | \  ..., \ \Prr: (l, x := d;  \sll),  ... > 
		\end{array}$$
		\caption{Operational semantics involving wrappers, where $[E]= \mathit{eval}(e,(a \# l))$ and $\mathit{nfId} = \mathit{new}(o,n)$.}
		\label{operational-semantics-wrapper}
\end{figure}}

The \textbf{invc-w'} rule represents a wrapper
and an incoming invocation message toward the object ($o$).
The notation $\varLambda [m,i] $ indicates the level of the
$i$th formal parameter of the method $m$ as declared in the class.
If the security level of each actual parameter ($\dl_{i}$)
is less than or equal to the security level of the corresponding formal parameter,
then the wrapper allows the message to go through 
and adds it to its configuration inside.
The \textbf{invc-q} equation stores an invocation message toward an object in the corresponding queue for later processing.
This equation
has priority over the rule of sending out an invocation message from a wrapper (if both rules apply),
since in Maude, an equation has priority over a rule.
The \textbf{w-fut} rule represents an unresolved future
and a corresponding completion message, which makes the future resolved.
On the right-hand-side, if the security level of the completion message ($\lev$) is high,
the future becomes
wrapped.
If the future value is low, there is no need to protect it
by a wrapper.
The wrapper has the same identity and level as the future.
The \textbf{w-get} rule represents a wrapped future and
an object which wants to get the future value. 
If the security level of the object ($\lev$) asking for the value is smaller
than the wrapper ($\lev'$),
then the wrapper sends an error value; otherwise, the object gets the future value ($d$).
The \textbf{fut-get} rule shows that an object gets the value from an unwrapped future without security checking.
\subsection{Non-interference}\label{proof}
We now sketch a proof 
that our security framework satisfies
noninterference.
 We say that two states of an 
object  are \emph{low equivalent} if
they agree on the low values of the variables (in $a\#l$),
while values of high variables may differ.
\begin{definition}
	\underline{Global non-interference} 
	means that 	
		for any two executions 
	with corresponding objects and futures, 
    and for  any corresponding objects
  with the same non-idle remaining statement lists ($sl$) and $\pcs$ stacks,
we have that the outputs of the two objects are low equivalent
(modulo the differences in object and future identities)
as long as they  consume 
the same low equivalent inputs. %
		\label{glob-def}
      \end{definition}
            \begin{definition}\underline{Local non-interference}
means that in two executions of an object 
with 
consumption of 
low equivalent inputs,
the states are low equivalent.
\label{L-Non-in}
\end{definition}
We first prove that each object is \textit{locally
  deterministic}, in the sense that
the next state of a statement (other than idle) is deterministic, i.e.,  depending only on
the prestate and any input.
     The only source of non-determinism is
      the independent speed of the objects, which means that 
      the ordering in the messages queues is in general non-deterministic.
	Only 
idle states may give rise to non-determinism. 

\begin{lemma}
	In our security model, each object is locally deterministic. 
	\label{local-deter}
\end{lemma}
\begin{proof}
	According to our operational semantics, for each statement 
	(other than idle) there is only one rule to apply, and
	for an \kw{if} statement there are two rules, but the choice is given by testing the security level of the guard.
	There is no interleaving of processes inside an object as well.
	\qed
\end{proof}
\begin{definition}\underline{Low-to-low determinism} means that
any low part of a state or output resulting from a statement is
determined by the low part of the prestate.
\end{definition}
\begin{lemma}
	In our security model, each object is low-to-low deterministic.
	\label{low-to-low}
      \end{lemma}
      \begin{proof} This can be proved by case analysis on the statements.
For an \kw{if} with a high test, there is no low state change
(resulting in low tags)
nor low output in 
any branch. For an \kw{if} with low test, the choice of branch is 
given by the low part of the prestate. The  cases for 
the  other statements 
are straightforward.
	\qed\end{proof}
\begin{theorem}
	Our security model 
	guarantees 
	local  and global non-interference,
and an attacker will only receive low information.
\end{theorem}
\begin{proof}
Local non-interference can be proved by induction of the number of execution steps
considering two executions of an object.
The low part of each state and the low outputs
must be the same  by the two previous lemmas. 
Global non-interference can be proved by induction on the number of steps
considering two executions. It follows by
local non-interference 
for all objects.
%
Since an attacker is a low object,
the wrappers will prevent it to receive high inputs, and the low inputs will be the same. 
	\qed\end{proof}
\vspace{-2mm}
\section{Related work}\label{related-work}
A static and class-wise information-flow analysis has been suggested for
Creol without futures 
by Owe and Ramezanifarkhani
in~\cite{oweramezani17dpm}.
The authors proved soundness and a non-interference property
in object interactions based on an operational semantics.
In contrast to the present paper, futures are not considered.
Our approach is a dynamic technique which is more permissive and
precise.
To reduce 
the run-time overhead, we combine it with a sound static analysis, 
as
the one 
in~\cite{oweramezani17dpm}.

In a paper by Attali et al.~\cite{attali2007secured}, a dynamic information-flow control approach is performed for the ASP language.
Security levels are assigned to activities and communicated data (an activity includes an active object and several passive objects
controlled by one thread).
The security levels do not change when they are assigned.
Dynamic checks are performed at
activity creations, requests, and replies.
Since future references are not confidential, they are passed between activities without dynamic checking, but
getting a future value is checked by
a reply transmission rule.
In~\cite{attali2007secured}, the security model guarantees data confidentiality for multi-level security (MLS) systems, which means that
an entity is allowed to access only the information that it is allowed to access.
Our approach adds flow-sensitivity, which allows security levels of variables to change inside an object.
It makes our approach more permissive and
a wrapper deals with run-time security levels.
In addition to enforcing the non-interference property in object
interactions, our approach guarantees that an object will be given
access only to the information that it is allowed to handle.

 Russo and Sabelfeld~\cite{russo2010dynamic} prove that a sound flow-sensitive dynamic information-flow enforcement is more permissive than static analysis.
Magazinius et al.~\cite{magazinius2010safe} use the notion of wrappers to control the behavior of
JavaScript programs and enforce security policies to protect web pages from malicious codes.
The security policies avoid web browser vulnerabilities and protect web pages.
A policy specifies under what conditions a page
may perform a specific action, and a wrapper grants, rejects, or modifies these actions. 
Moreover, the notion of wrappers has
been developed for the safety of objects~\cite{owe2009wrap},
letting an object wrapper control
which actions should be taken for any input/output. We here exploit wrappers for dealing
with information security, by extending the run-time system with security 
levels,
and applying dynamic checking for securing object interactions.

\vspace{-2mm}
\section{Conclusion}
We have proposed a  framework 
for enforcing 
secure information-flow and non-interference in active object languages
based on the notion of security wrappers.
We have considered a high-level core language 
supporting asynchronous calls and futures.
In our model, due to encapsulation, there is no need for
information-flow restrictions inside an object.
Wrappers perform security checks for object interactions 
(with methods and futures)
at run-time.
Furthermore, wrappers control the access to the futures with confidential values. 
Security rules of wrappers are defined based on security levels
assigned to objects, actual 
parameters, fields and local variables.
Inside an object, 
the security levels of variables might change at run-time due to
flow-sensitivity.
In the setting where concurrent objects communicate confidential or
non-confidential information, 
wrappers on unsafe objects and future components
protect exchange of  confidential 
values. 
Wrappers on objects protect outgoing method calls and prevent leakage of information through outgoing parameters.
The wrappers are created by the run-time system without the involved parties being aware of it.
The notion of security wrappers enable 
dynamic enforcement of avoidance  of information flow leakage.
We define non-interference  and outline a proof of  it.
By combining results from static analysis,
we can improve run-time efficiency by avoiding wrappers
when 
 they are superfluous according to 
the  over-approximation of levels  given by
the static analysis.

%% file: appendix.tex
\subsection{Appendix}\label{appendix}
\begin{theorem}[No leackage from a wrapped future]\label{sec-fut}
	A high future value
	will not be passed to a low object.
\end{theorem}%
\begin{proof}
	Since each call has a unique future identity,
	there is at most one write operation on a future component,
	and by the rule \textbf{w-fut} in Fig.~\ref{operational-semantics-wrapper},
	the future component is placed in a wrapper if the 
	future value is high.
	And this wrapper is not removed by any rule.
	Any get operation of such a future
	is handled by rule \textbf{w-get} in Fig.~\ref{operational-semantics-wrapper},
	and this rule will not pass such a high future value to a low
        object.
\qed
\end{proof} 
\begin{theorem}[No high message to a low level object]\label{sec-mes}
	A message with high level tag is not sent to a low level object.
\end{theorem}%
\begin{proof}
	According to the \textbf{w-invc} rule in Fig.~\ref{operational-semantics-wrapper},
	an invocation message is only sent out of a wrapper if the level of the destination object is higher than the level of the message. Otherwise, the invocation message is deleted from inside the wrapper.
	This wrapper is not removed by any rule.
	For a completion message, if
	the level of the completion message is high, then the corresponding future becomes wrapped (rule \textbf{w-fut} in Fig.~\ref{operational-semantics-wrapper}).
	According to the Theorem~\ref{sec-fut}, there is no leakage from a wrapped future.
	\qed
\end{proof} 
We define low outputs generated from an object as:
\begin{definition}(Low outputs)
	The low output of a call statement performed by an
	object is 
	the generated low invocation message. 
	An invocation message is low if it does not contain any actual parameter with high security level, and it is not generated in a high context where $\pc=H$ (rule \textbf{call'} in Fig.~\ref{operational-semantics-flowsen}).
	The low output of a new statement is given by the actual class
	parameters.
	The observable output of a return statement
	is the low completion message. 
	A completion message is low if the return value is low, and it is not generated in a high context (rule \textbf{return} in Fig.~\ref{operational-semantics-flowsen}).
\end{definition}
Non-interference means that
in two executions, if we change the high inputs but not the low ones, the resulting low outputs for the attacker must be the same.
In our security approach, an attacker does not get any high message because of the security wrappers (proved by Theorem~\ref{sec-mes}),
but it can receive low messages or low outputs from other objects.
In order to prove non-interference, it suffices to show that the low outputs of the two executions are low equivalent, i.e., they agree on low variables, and low outputs are not resulted from high information.

By Lemma~\ref{local-deter}
we prove that each object is \textit{locally
deterministic}, in the sense that
the next state of a statement (other than idle) only depends on
the prestate and any input.
We prove by Lemma~\ref{low-to-low} that objects are low-to-low determinism, where the low part of the next state of a statement is determined by the low part of
the prestate or any input.
One interesting case for Lemma~\ref{low-to-low} could be that when inside an object $\pcs\neq \mathit{emp}$ (or $\pc=H$), then after each statement, all the state changes are high, and the object produces no low outputs. It is proved in Lemma~\ref{behav}.

In order to show the behavior of an object when $\pcs\neq \mathit{emp}$ and non-interference, we let $R$ denote a run-time configuration,
and
$R/o$ the projection of the local run-time configuration of  object $o$.
To conserve space, $R/o$ is a compact notation of a flow-sensitive object component (\textit{flowsen-obj}) in Fig.~\ref{fig:object-model}, and it contains attributes like the identity ($o$), the current statement that the object is about to perform ($\mathit{sl}$), the \emph{current state} of the object ($\sigma=a\#l$), the stack of $\pcs$, and the rest is indicated by [...]. 
For example, $R/o= <o  |  \mathit{sl}, \sigma,..., \pcs>$ is $o$'s configuration, where $\mathit{sl}$ is a list of statements to be performed by the object, $\sigma$ is the current state of the object, a map composition of the attribute state $a$ and local state $l$ ($\sigma=a\#l$), and $\pcs$ is the stack of context security level.
According to the operational semantics in Fig.~\ref{operational-semantics-flowsen}, variables are stored in the sate of an object with their values and security levels, where $\sigma^{\mathit{se}}$ denotes the \textit{security level environment}, mapping variables to security levels,
and $\sigma^{m}$ denotes the object \textit{memory}, mapping variables to values.
One step of execution is denoted by:
$<o  | s, \sigma,...,\pcs>\longrightarrow_{\alpha }<o | s', \sigma',...,\pcs'>$,
where $\alpha$ is the possible produced output, $\sigma$ is the prestate state, and $\sigma'$ is the next state.
A sequence of execution steps inside an object is denoted by $\longrightarrow_{\alpha}^{*}$,
where $\alpha$ is a (possible empty) list of produced concatenated outputs (including outgoing invocation and completion messages, created futures, and objects).
 The relation $\alpha_{1}<\alpha_{2}$ means that $\alpha_{1}$ is a strict prefix of $\alpha_{2}$, and
 $\mathit{low}(\alpha)$ denotes the projection of low outputs from $\alpha$.

The next lemma describes the behavior of an active flow-sensitive object when $\pcs\neq\mathit{emp}$ (which corresponds to $\pc=H$). We say that two states, inputs, or outputs of an 
object  are low equivalent ($\approx_{L}$) if
they agree on the low values of the variables,
while variables with high values may differ.
%
\begin{lemma} \label{behav}
	Given a statement $s$ that contains no \textit{endif} instruction, $\pcs\neq \mathit{emp}$, and $<o| s, \sigma, ...,\mathit{pcs}> \longrightarrow_{\alpha}^{*} <o |\mathit{idle},\sigma', ..., \mathit{pcs'}>$ then:
	\begin{enumerate}[label=(\roman*)]
		\item $pcs'=pcs $
		\item $\sigma'=\sigma \sqcup \mathit{update}_{\mathit{H}}(s)$
		\item $\sigma'\approx_{L} \sigma$ (The post and prestate are low equivalent.)
		\item $\mathit{low}(\alpha)=\emptyset$
	\end{enumerate}
	where $\sigma \sqcup \mathit{update}_{\mathit{H}}(s)$ raises the security levels of variables appearing on the left-hand-side of assignments in the statement $s$ and updates them in the object's state. 
\end{lemma}

\begin{proof}
	by induction on $\longrightarrow^{*}$, we show it with some case analysis and the operational semantics in Fig.~\ref{operational-semantics-flowsen} (where in the rules, as mentioned before, fields that are not changed on the right-hand-side are denoted by ....):
	\begin{itemize}
		\item[--] $x:=e;$\\ According to the \textbf{assign} rule, item (\textit{i}) holds since an assignment does not change the $pcs$ stack.
		Item (\textit{ii}) holds since $\pc=H$, and $\sigma'= \sigma[x \mapsto [e]_{\mathit{level}([e])\sqcup pc}] =\sigma [x\mapsto [e]_\mathit{H}]\Rightarrow \sigma'=\sigma \sqcup \mathit{update}_{H}(x:=e) $. Item (\textit{iii}) holds since there is no low state change resulting in a low tag, and the values of low variables  in the post state $\sigma'$ and  the prestate $\sigma$ are the same. An assignment produces no output, thus item (\textit{iv}) holds.
		\item[--] \kw{return} $d$;\\
		According to the \textbf{return} rule, item (\textit{i}) holds since it does not change the $\mathit{pcs}$. $\mathit{update_{H}(\kw{return}\ d)}=\emptyset$, and item (\textit{ii}) holds.
		 This statement does not change the object's state, so the post and pre-state are low equivalent, and item (\textit{iii}) holds .
		Since $\pc=H$, the \textbf{return} statement generates a completion message with a high security level, and item (\textit{iv}) holds .
		\item[--] $q! o'.m(\many{e})$; \\According to the \textbf{call} rule, this statement is reduced to $q:=\newfId;! o'.m(\many{e})$.
		The first statement is an assignment; therefore, all the items hold, and we have $\sigma'=\sigma\sqcup\mathit{update}_{\mathit{H}}(q:=\newfId)$.
		For the second statement, according to the \textbf{call'} rule, item (\textit{i}) holds since it does not change the $\mathit{pcs}$.
		A call statement does not change the object's state, so the post and pre-state are low equivalent, and items (\textit{ii, iii}) hold, where $\mathit{update_{H}(! o'.m(\many{e}))}=\emptyset$.
		Item (\textit{iv}) holds, since  $\pc=H$, and this statement creates an invocation message with a high security level.
		\item[--] $\kw{if}\ e \ \kw{th}\ s_1 \ [\kw{el}\ s_2]\kw{fi}$;\\ 
		Since $\pcs\neq \mathit{emp}$, the \textbf{if-high} rule is applied.
		According to this rule, $H$ is pushed to the $\pcs$; however, at the joint point, the added level is removed from the $\pcs$ ($\pcs.\mathit{pop()}$) by the \textbf{endif} rule.
		Therefore, at the end of an \textbf{if} structure, we have $\pcs'=\pcs$, and item (\textit{i}) holds.
		The taken branch is evaluated in a high context, 
		and the security levels of variables in the untaken branch are raised by $H$ in the \textbf{endif} rule. 
		Therefore, at the end of an if structure, item (\textit{ii}) holds, where $\sigma'=\sigma\sqcup\mathit{update}_{\mathit{H}}(s_1, s_2)$.
		Item (\textit{iii}) holds, since the \textbf{if-high} rule only results in high state changes, so $\sigma'\approx_{L} \sigma$.
		Item (\textit{iv}) holds since $\pc=H$; therefore, any output generated in either of the branches has a high security level.
%
%
%
		\item[--] $x:=\kw{new}_{lev}C'(\many{e})$; \\According to the \textbf{new} rule, this statement is reduced to $x:=\mathit{newId}; !x.\mathit{init}(\many{e})$, where for the first statement, all the items hold similar to an assignment. For the second statement, similar to a call (when $\pc=H$) all the items hold.
		\item[--] $\mathit{fId}?(x)$ (a get statement on the future $\mathit{fId}$);\\ According to the \textbf{fut-get} rule in Fig.~\ref{operational-semantics-wrapper}, this statement reduces to $x:=d$, where $d$ is the future value. Items (\textit{i, ii, iii}) hold similar to an assignment. Item (\textit{iv}) holds since it produces no output.
		This statement might raise an exception in case the object's level is low and the the future value is high as shown in the \textbf{w-get} in Fig.~\ref{operational-semantics-wrapper}.
	\end{itemize}
	\qed\end{proof}
Next we formalize and prove local non-interference, defined in Definition~\ref{L-Non-in}.
 $o_{1}$ and $o_{2}$ are the corresponding objects in two executions of an object.
\begin{lemma}(Local Non-Interference for Objects). \label{non-inter-formal}
	\begin{align*}
	&  <o_{1}|\mathit{sl}, \sigma_{1},..., \pcs>\longrightarrow_{\alpha_{1}}^{*} 
	<o_{1} |\mathit{sl'},\sigma_{1}',..., \pcs'>  \wedge \\ 
	& <o_{2} | \mathit{sl}, \sigma_{2},...,\pcs>\longrightarrow_{\alpha_{2}}^{*} <o_{2} |\mathit{sl''}, \sigma'_{2},...,\pcs''> \wedge \ \sigma_{1}\approx_{L}\sigma_{2} \Rightarrow\\ 
	& \sigma'_{1}\approx_{L}\sigma'_{2}\ \wedge \ \pcs''=\pcs' \ \wedge \ \mathit{low}(\alpha_{2} )= \mathit{low}(\alpha_{1} ) 
	\end{align*}
\end{lemma}
\begin{proof}
	It can be proved by induction of the number of execution steps
	considering two executions of an object with wrappers and corresponding futures.
	The low part of each state and the low outputs
	must be the same  by Lemmas  \ref{local-deter}, \ref{low-to-low}.
	It is obvious that according to Lemma~\ref{behav}, if $\pcs\neq \mathit{emp}$ then there is no low state change or output in both objects. 
	The proof is shown for some statements when $\pcs=\mathit{emp}$.
	\begin{itemize}
		\item[--] $\kw{return}\ d$;
		{\small	\begin{align*}
			<o_{1}| \kw{return}\ d, \sigma_{1},..., \pcs>&\longrightarrow_{\alpha_{1}} <o_{1}| \mathit{idle},\sigma'_{1},..., \pcs'> \\
			<o_{2}| \kw{return}\ d,\sigma_{2},..., \pcs>&\longrightarrow_{\alpha_{2}} <o_{2}| \mathit{idle},\sigma'_{2},..., \pcs''> 
			\end{align*}}
		Since a return statement does not change an object's state or the $\pcs$, we have $\sigma_{1}=\sigma'_{1}\approx_{L}\sigma'_{2}=\sigma_{2}$ and $\pcs''=\pcs'$.  
		If $\mathit{level}(d)=\mathit{L}$, then $d$ has the same value in both states because $\sigma_{1}\approx_{L}\sigma_{2}$, i.e., $\sigma_{1}^{m}(d)=\sigma_{2}^{m}(d)$ ($\sigma^{m}(d)$ returns the value of $d$ stored in $\sigma$).
		Correspondingly, the objects produce low completion messages, where $ \mathit{low}(\alpha_{2} )= \mathit{low}(\alpha_{1} )$.
		Moreover, the corresponding futures also become low equivalent.
		If $\mathit{level}(d)=H$, then both objects produce high completion messages, where $\mathit{low}(\alpha_{2} )= \mathit{low}(\alpha_{1} )=\emptyset$.
		The corresponding futures become high and wrapped correspondingly.
		\item [--] $\kw{if} \ e \ \kw{th} \ s_{1}\ \kw{el} \ s_{2}$;\\ 
		Depending on the level of $e$, either the \textbf{if-low} rule or the \textbf{if-high} rule in Fig.~\ref{operational-semantics-flowsen} is applied.
		If $\mathit{level}(e)=L$, then $\sigma_{1}^{m}(e)=\sigma_{2}^{m}(e)$, and both objects take the same branch, thus $\sigma'_{1}\approx_{L}\sigma'_{2}$ holds.
		The \textbf{if-low} rule does not change the $\pcs$, and  $\pcs''=\pcs'$ holds.
		If $\mathit{level}(e)=H$, the \textbf{if-high} rule pushes $H$ to the $\pcs$, resulting $\pc=H$.
		Since $\mathit{level}(e)=H$, we have $\sigma_{1}^{m}(e)\neq \sigma_{2}^{m}(e)$, and objects take different branches
		but produce no low outputs, where $\mathit{low}(\alpha_{2} )= \mathit{low}(\alpha_{1} )=\emptyset$.
		It is proved by Lemma~\ref{behav} that when $\pc=H$, an object does not produce a low output.
		Moreover, we have $\sigma'_{1}\approx_{L}\sigma'_{2}$, it follows from the fact that in the \textbf{if-high} and \textbf{endif} rules, the security levels of variables appearing on the left-hand-side of assignments in both branches are raised to high. Thus states of both objects remain low equivalent after updating, i.e.,
		$\sigma_1\approx_L \sigma_2 \Rightarrow \sigma_1 \sqcup \mathit{update}_{H}(s_{1},s_{2})\approx_{L} \sigma_{2} \sqcup \mathit{update}_{H}(s_{1},s_{2}) \Rightarrow \sigma'_{1}\approx_{L}\sigma'_{2}$.
		According to the \textbf{endif} rule, the $\pcs$ stacks of objects return back to the previous context level, and $\pcs''=\pcs'=\pcs$ holds. 
		\item [--] $!o'.m(\many{e});s'$;
		{\small	\begin{align*}
			<o_{1}|!o'.m(\many{e});s', \sigma_{1},...,\pcs>&\longrightarrow_{\alpha_{1}} <o_{1}| s',\sigma'_{1},...,\pcs'>\ \\
			<o_{2}| !o'.m(\many{e});s', \sigma_{2},...,\pcs>&\longrightarrow_{\alpha_{2}} <o_{2}| s', \sigma'_{2},...,\pcs''>
			\end{align*}}
According to the \textbf{call'} rule in Fig.~\ref{operational-semantics-flowsen}, since a call statement does not change an object' state and the $\pcs$, we have $\sigma_{1}=\sigma'_{1}\approx_{L}\sigma'_{2}=\sigma_{2}$ and $\pcs''=\pcs'=\pcs$.
		 If $\mathit{level}([\many{e}])=L$, it means that all the actual parameters in $[\many{e}]$ are low level, and we have $\sigma_{1}^{m}(\many{e})=\sigma_{2}^{m}(\many{e})$.
		Therefore, the objects produce low invocation messages, where $\mathit{low}(\alpha_{2} )= \mathit{low}(\alpha_{1} )$. If $\mathit{level}([\many{e}])=H$, then the objects produce high invocation messages, and
		according to the \textbf{w-invc} rule in Fig.~\ref{operational-semantics-wrapper}, wrappers do not send these high messages to a low level object.
		\item[--] $x:=\kw{new}_{lev}C'(\many{e})$;\\
		According to the \textbf{new} rule in Fig.~\ref{operational-semantics-flowsen}, this statement is reduced to $x:=\mathit{newId}; !x.\mathit{init}(\many{e})$.
		The first and the second statements do not change the $\pcs$, thus $\pcs''=\pcs'=\pcs$ holds.
		The first assignment assigns a new object identity to $x$, thus the states of objects remain low equivalent  $\sigma'_{1}\approx_{L}\sigma'_{2}$.
		In the second statement, a call does not change the states of the objects, thus $\sigma'_{1}\approx_{L}\sigma'_{2}$ holds. If $\mathit{level}(\many{e})=L$, then the objects produce the same low invocation messages and $\alpha_{1}=\alpha_{2}$. Otherwise, the invocation messages are high, and if the new objects' levels are low, then the wrappers do not send these messages to the new objects.  
	\end{itemize}
	\qed \end{proof}

\begin{theorem}
	Our security model 
	guarantees 
	local non-interference.
	\label{theo-ln}
\end{theorem}
\begin{proof}
	By induction of the number of execution steps and lemma~\ref{non-inter-formal}.
	\qed\end{proof}
Next we formalize and prove global non-interference, defined in the definition~\ref{glob-def}.
At the global view, due to non-determinism, two
corresponding objects might receive input messages in different orders.
Thus, it is assumed that the executions start when both objects start with the same low equivalent input messages and the queries on futures are also low equivalent.
In order to describe non-determinism, we need to consider the set of all possible executions, where two corresponding objects will take possible messages out of their queues, for which we use an existential quantifier.
The notation $R/o_{1}\approx R/o_{2}$ indicates two corresponding objects in the two executions.
The notation $R/o \nRightarrow_{L}$ indicates that this configuration does not produce any low output.
\begin{lemma}(Global Non-Interference for Objects). \label{glob-formal}
\begin{align*}
&<o_{1}|s, \sigma_{1},...,\pcs>\approx<o_{2} | s, \sigma_{2},...,\pcs> \ \wedge \ \sigma_{1}\approx\sigma_{2}  \\
& \wedge <o_{1}|s, \sigma_{1},...,\pcs>\longrightarrow_{\alpha_{1}}^{*} 
<o_{1} |s',\sigma_{1}',...,\pcs'>\ \Rightarrow \ \exists \ s'', \sigma'_{2}, \pcs'', \alpha_{2} \ .\\
& 
 <o_{2} | s, \sigma_{2},...,\pcs>\longrightarrow_{\alpha_{2}}^{*} <o_{2} |s'', \sigma'_{2},...,\pcs''>\wedge \ \sigma'_{1}\approx_{L}\sigma'_{2}\ \wedge \ \pcs''=\pcs'
 \\
&    \wedge \ \mathit{low}(\alpha_{2} )\leq \mathit{low}(\alpha_{1} )\ \wedge \
 \mathit{if}\ \mathit{low}(\alpha_{2} )< \mathit{low}(\alpha_{1} ) \ \mathit{then}  <o_{2} |s'', \sigma'_{2},...>\nRightarrow_{L}
\end{align*}
\end{lemma}
\begin{proof}
By induction on $\xlongrightarrow{H}$* on both executions of the two objects with wrappers and corresponding futures, and Lemmas~\ref{non-inter-formal} and \ref{behav}. 
\qed \end{proof}
\begin{theorem}(Non-interference)
	Our security model using
	the wrapper mechanism
	guarantees
	global non-interference, and an attacker will only receive low information.
\end{theorem}
\begin{proof}
	%
	By Lemma~\ref{glob-formal} and Theorems~\ref{sec-fut} and~\ref{sec-mes}.
	\qed\end{proof}